\documentclass[twocolumn,vr_35]{ieeetran}
\usepackage{etoolbox}
\makeatletter
\patchcmd{\@makecaption}
  {\scshape}
  {}
  {}
  {}
\makeatother
 \usepackage{amsmath,amssymb}
 \usepackage{subfigure}
 \usepackage{graphicx,graphics,color,psfrag}
 \usepackage{cite,balance}
 \usepackage{caption}
 \allowdisplaybreaks
 \usepackage{algorithm}
 \usepackage{accents}
 \usepackage{amsthm}
 \usepackage{bm}
 \usepackage{algorithmic}
 \usepackage[english]{babel}
 \usepackage{multirow}
 \usepackage{enumerate}
 \usepackage{cases}
 \usepackage{stfloats}
 \usepackage{dsfont}
 \usepackage{color,soul}
 \usepackage{amsfonts}
 \usepackage{cite,graphicx,amsmath,amssymb}
 \usepackage{subfigure}
 \usepackage{fancyhdr}
 \usepackage{hhline}
 \usepackage{graphicx,graphics}
 \usepackage{array,color}
 \usepackage{amsmath}
\usepackage{float}
\usepackage{amssymb}
\usepackage{amsmath}
\usepackage{amsthm}
\usepackage{amsfonts}
\usepackage{graphicx}
\usepackage{algorithm}
\usepackage{algorithmic}
\usepackage{epstopdf}
\usepackage{cite}
\usepackage{amsmath,bm}
\usepackage{subfigure}
\usepackage{graphicx}
\usepackage{color}
\usepackage{graphicx}
\usepackage{calc}
\usepackage{caption}
\usepackage{makecell}

\newtheorem{theorem}{\textbf{Theorem}}

\newtheorem{lemma}{\textbf{Lemma}}
\newtheorem{remark}{\textbf{Remark}}
\newtheorem{Prob}{\textbf{Problem}}

\newtheorem{property}{Property}
\begin{document}
\title{Communications, Caching and Computing for Mobile Virtual Reality: Modeling and Tradeoff}
\author{Yaping Sun, Zhiyong Chen, Meixia Tao and Hui Liu\thanks{The paper was presented in part at IEEE ICC 2018 \cite{Sunvr}.

Y. Sun and M. Tao are with the Department of Electronic Engineering, Shanghai Jiao Tong University, Shanghai 200240, China (e-mail: yapingsun@sjtu.edu.cn; mxtao@sjtu.edu.cn).

Z. Chen and H. Liu are with the Cooperative Medianent Innovation Center, Shanghai Jiao Tong University, Shanghai 200240, China, and also with the Shanghai Key Laboratory of Digital Media Processing and Transmissions, Shanghai 200240, China (e-mail: zhiyongchen@sjtu.edu.cn; huiliu@sjtu.edu.cn).}}

\maketitle
\begin{abstract}
Virtual reality (VR) over wireless is emerging as an important use case of 5G networks. Immersive VR experience requires the delivery of huge data at ultra-low latency, thus demanding ultra-high transmission rate. This challenge can be largely addressed by the recent network
architecture known as mobile edge computing (MEC), which enables caching and computing
capabilities at the edge of wireless networks. This paper presents a novel MEC-based mobile
VR delivery framework that is able to cache parts of the field of views (FOVs) in advance and run certain post-processing procedures at the mobile VR device. To optimize resource allocation at the mobile VR device, we formulate a joint caching and computing decision problem to minimize the average required transmission rate while meeting a given latency constraint. When FOVs are homogeneous, we obtain a closed-form expression for the optimal joint policy which reveals interesting communications-caching-computing tradeoffs. When FOVs are heterogeneous, we obtain a local optima of the problem by transforming it into a linearly constrained indefinite quadratic problem then applying concave convex procedure. Numerical results demonstrate great promises of the proposed mobile VR delivery framework in saving communication bandwidth while meeting low latency requirement.
\end{abstract}

\section{Introduction}\label{I}
\subsection{Motivation}
Virtual reality (VR) over wireless, namely mobile VR delivery, is emerging as an important use case of 5G and beyond networks, due to its ability to generate an immersive experience at the full fidelity of human perception \cite{Sunvr,burst,whitepaper}. A recent market report forecasts that the data consumption from mobile VR devices (smartphone-based or standalone) will grow by over $650\%$ in the next 4 years (2017-2021)\cite{VRmarket}. Immersive VR experience requires the delivery of massive amount of data (on the order of \textit{Gigabyte}) at ultra-low latency (less than \emph{$20$ ms}), thus demanding
ultra-high transmission rate and leading to the wireless bandwidth bottleneck problem \cite{whitepaper}.  

In order to tackle the challenge, the recent network architecture concept known as mobile edge computing (MEC), which enables caching and computing capabilities at the edge of wireless networks, is envisioned as one of the key enablers for mobile VR delivery \cite{e}. The main idea of the MEC network is to  cache strategic contents in advance and compute certain post-processing procedures on demand at the mobile edge network, thereby reducing traffic load as well as response time. Thus, in this paper, we aim to investigate the mobile VR delivery using MEC network architecture and find out how to make the best use of the caching and computing capabilities of the MEC network to minimize the bandwidth requirement for mobile VR delivery while satisfying the stringent latency constraint.


%
\subsection{Our Contributions}
\begin{figure}[t]
\begin{center}
 \includegraphics[width=9cm]{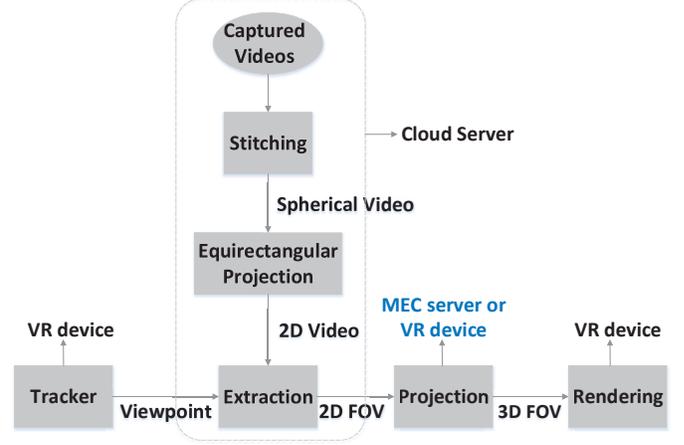}
\end{center}
 \caption{\small{A typical framework of $360^\circ$ VR video producing \cite{Simone}.}}
\label{Architecture}
\end{figure}
To illustrate the problem at hand, we first analyze a typical $360^\circ$ VR video producing framework \cite{Simone}, as shown in Fig. \ref{Architecture}: 
i) \emph{Stitching}, which obtains a spherical video by stitching the videos captured by a multi-camera array; ii) \emph{Equirectangular projection}, which obtains 2-dimensional (2D) video by unfolding the obtained spherical video; iii) \emph{Extraction}, which extracts the 2D video to obtain 2D field of view (FOV) of the viewpoint captured by the \emph{tracker} at the mobile VR device;
iv) \emph{Projection}, which projects 2D FOV into 3D FOV; 
v) \emph{Rendering}, which renders the obtained 3D FOV onto the display of the mobile VR device.

 We propose the following method to realize the above mentioned framework within the MEC network illustrated in Fig.~\ref{model}, which consists of one MEC server (e.g., base station) and one mobile VR device, both with certain caching and computing capabilities. First, without doubt, the tracker and rendering components must be computed at the mobile VR device. Secondly, we assume that the first three pre-processing procedures including stitching, equirectangular projection and extraction components are computed offline at the cloud server. Considering that such three components require the entire $360^\circ$ video as inputs, realizing them at the cloud server can release both MEC server and mobile VR device from heavy computation process as well as alleviate the traffic burden within the wireless network. Then, 2D FOVs of all the viewpoints extracted at the cloud server can be cached at the MEC server in advance, thereby reducing the traffic burdern on the backhaul link and also the response latency. Moreover, the projection component can be computed offline at the MEC server, and thus 3D FOVs of all the viewpoints can also be proactively cached at the MEC server.



A key observation is that the projection component can be offloaded from the MEC server to the mobile VR device due to its low computational complexity \cite{3D} and the increasing computing capability of the mobile VR device \cite{e}. Specifically, compared with downloading the requested 3D FOV from the MEC server, named as MEC computing, downloading 2D FOV from the MEC server and then computing the projection at the mobile VR device can reduce at least half of the traffic load on the wireless link. This is due to the fact that in order to create a stereoscopic vision, the projection component has to be computed twice (one for each eye) to obtain two slightly differing images \cite{3D}, and hence the data size of  3D FOV is at least twice larger than that of the 2D FOV. 
However, computing at the mobile VR device incurs additional computation latency. Thus, \textit{the computing policy}, i.e., whether to compute the projection at the mobile VR device or not, requires careful design. In addition, caching capability of the mobile VR device can be utilized to store 2D FOVs or 3D FOVs of some viewpoints for future requests. Specifically,  compared with caching a 2D FOV,
caching 3D FOV can help reduce both latency and energy consumption, since the 3D FOV request can be directly served  and without the need of transmission and computing. However, 3D FOV caching consumes at least twice larger caching resource at the mobile VR device than 2D FOV caching. Thus, \textit{the caching policy}, i.e., caching  2D FOVs or 3D FOVs at the mobile VR device, also requires careful design.

Main contributions of this paper are summarized as follows.
\begin{itemize}
\item \emph{A novel MEC-based framework for mobile VR delivery:} We propose a realization method for mobile VR delivery, as mentioned above. This method allows the pre-processing procedures computed at the cloud server and post-processing procedure, i.e., projection component, computed at the MEC server or the mobile VR device, thereby significantly reducing the transmission data within the wireless network as well as required latency.
\item \emph{Optimal joint caching and computing policy:} Based on the proposed realization method, when FOVs are homogeneous, we formulate joint caching and computing decision problem to minimize the average transmission rate, under the latency, local cache size and average energy consumption constraints. By analyzing the optimal properties and solving several linear programming problems, a closed-form expression for the optimal joint policy is obtained and provides useful guidelines for network designers on how to make the best use of caching and computing capabilities of the mobile VR device.
\item \emph{Communications-caching-computing tradeoff:}
Based on the optimal joint policy, we derive the minimum required transmission rate and theoretically illustrate the communications-caching-computing (3C) tradeoff. Analytical results show that compared with MEC computing, the transmission rate gain under the optimal joint policy comes from the following three aspects: local 3D caching, local computing with local 2D caching and local computing without local caching. We theoretically reveal that such three gains can be exploited opportunistically in different caching and computing capability regimes. For example, when the computation frequency at the mobile VR device is relatively small, there is no local computing gain without local caching, and transmission rate gain comes from local 3D caching and local computing with local 2D caching. In addition, caching resource at the mobile VR device is exploited more efficiently joint with computing resource, and vice versa; when the computation frequency is large enough, the gain comes from local 3D caching and local computing with/without caching coherently. The power efficiency of the mobile VR device is also shown to play an important role in the transmission rate via determining the local computing gain directly. More details can be seen in Section IV.



\item  \emph{Heterogeneous scenario optimization:} We extend the joint caching and computing optimization problem to the scenario where FOVs are heterogeneous. In particular, we first show the NP-hardness of the joint policy optimization problem, and then obtain a local optima of the problem via transforming it into an equivalent linearly constrained indefinite quadratic problem (IQP) and using concave convex procedure (CCCP) \cite{localoptima}.  Numerical results demonstrate great promises of the proposed mobile VR delivery framework in saving communication bandwidth while meeting low latency requirement.
\end{itemize}

\subsection{Related Works}
Researchers in both academia and industry have made great efforts in order to achieve mobile VR delivery. First of all, at any given time, since each user only watches a portion of the $360^{\circ}$ VR video, the requested FOV is chosen to be transmitted instead of the entire panoramic video, thereby saving bandwidth significantly. Then, by knowing each user's FOV, multi-view and tile-based  video streaming have been investigated in \cite{tile1,tile2}. To further improve the quality of experience, motion-prediction-based transmission is also being studied based on dataset collected from real users \cite{qian,fixation,chen1,chen2}. However, \cite{3D,tile1,tile2,qian,fixation,chen1,chen2} mainly focus on the VR video-level design, and have not investigated the opportunities for mobile VR delivery potentially obtained via efficiently using the MEC network architecture. 

The opportunities for mobile VR delivery that can be potentially obtained via efficiently utilizing resources at MEC network, i.e., 3C, have been studied in \cite{e,erol,mohammed1,Simone,J,colla,bigdata,game,mohammed2}. Specifically, \cite{e,erol,mohammed1} envision joint computing and caching as the key enablers for mobile VR delivery and illustrate the potential gain  via simulation results. \cite{Simone} provides an explicit VR framework, based on which the insights on how to deliver $360^\circ$ video in mobile edge network are illustrated. However, \cite{e,erol,mohammed1,Simone} do not establish explicit theoretical formulation or propose any efficient algorithms. On the other hand, \cite{colla} proposes a collaborative cache allocation and computation offloading policy, where the MEC servers collaborate for executing computation tasks and data caching. \cite{bigdata} extends the results in \cite{colla} to a big data MEC network. \cite{game} proposes hybrid control algorithms in smart base stations  along with devised communication, caching, and computing techniques based on game theory. However, the joint caching and computing designs developed in \cite{colla,bigdata,game} do not exploit specific nature of VR delivery and look deeper into the VR delivery framework, and thus the performances  are limited. \cite{J} formulates an optimization framework for VR video delivery in a cache-enabled cooperative multi-cell network and explores the fundamental tradeoffs between caching, computing and communication for VR/AR applications. \cite{mohammed2} proposes joint  policy based on millimeter wave communication for interactive VR game applications.

\begin{figure}[t]
\begin{center}
 \includegraphics[width=9cm]{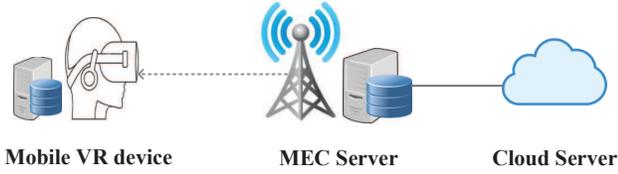}
\end{center}
 \caption{\small{MEC-based VR delivery model.}
 }
\label{model}
\end{figure}
It is worthy to note that all the works in \cite{colla,bigdata,game,J,mohammed2}  try to utilize the caching and computing resources at the MEC servers to alleviate the computation burdens at the mobile devices. However, as mentioned above, for the mobile VR delivery, computing at the MEC server may incur more transmission data since the computation results are generally larger than the inputs. Thus, in this paper, we focus on utilizing the caching and computing capabilities at the mobile VR device  to alleviate the communication burden on the wireless link and tackle the wireless bandwidth bottleneck problem. \cite{xiaoyang} exploits the caching and computing capabilities at the mobile VR device to minimize the traffic load over wireless link. However, \cite{xiaoyang} does not capture the specific nature of VR delivery and designs the optimal computation offloading policy based on the most popular caching policy, and thus the performance is limited. Therefore, the fundamental impacts of caching and computing resources at the mobile devices on the bandwidth requirement still have not been fully unleashed.


\subsection{Outline}
An outline of the remainder of the paper is as follows. Section II describes the system model for the MEC-based mobile VR delivery system under consideration. Section III formulates the joint policy optimization problem for the homogeneous scenario. Section IV obtains the optimal policy and 3C tradeoffs. Section V formulates the optimization problem  for the heterogeneous scenario and obtains the local optima via CCCP. Section VI concludes the paper.

\begin{table}[t]
\caption{Key Notations}
\begin{center}
\vspace{-2mm}
\newcommand{\tabincell}[2]{\begin{tabular}{@{}#1@{}}#2\end{tabular}}
\newcommand{\tl}[1]{\multicolumn{1}{l}{#1}} 
\renewcommand\arraystretch{1}
\setlength{\tabcolsep}{2mm}{
\begin{tabular}{|c!{\vrule width 0.9pt}l|}
\hline
Notation&Meaning\rule{0pt}{2mm}\\
\hhline{|=|=|}
\hline
$\mathcal{N},N$, $i$&  set of viewpoints,  number of viewpoints,  viewpoint index \\
\hline
$D^I,w,D^O,\tau$& \tabincell{l}{data size of 2D FOV, computation load, \\data size of 3D FOV, maximum tolerable service latency}\\
\hline
$C$, $\bar{E}$, $f_V$ &  \tabincell{l}{ cache size, average available energy, computation \\frequency at the mobile VR device}  \\
\hline
$f_S$ & computation frequency at the MEC server\\

\hline
$R_S$&  \tabincell{l}{the least required transmission rate  when the projection\\ is computed at the MEC server} \\
\hline
$R_V$ & \tabincell{l}{the least required transmission rate when the projection\\ is computed at the mobile device without  caching}\\
\hline
$c_i^I\in \{0,1\}$ & \tabincell{l}{$c_i^I = 1$ means that the 2D FOV of viewpoint $i$ is stored \\at the mobile VR device and not otherwise.} \\
\hline
$c_i^O\in \{0,1\}$ & \tabincell{l}{$c_i^O = 1$ means that the 3D FOV of viewpoint $i$ is stored \\at the mobile VR device and not otherwise.} \\
\hline
$d_i\in \{0,1\}$ & \tabincell{l}{$d_i=1$ means that projection is computed \\at the mobile VR device and not otherwise.}\\


\hline

\end{tabular}}
\vspace{-2mm}
\end{center}
\label{notation}
\end{table}

\section{System Model}

As illustrated in Fig.~\ref{model}, we consider a novel MEC-based mobile VR delivery system consisting of one MEC server and one mobile VR device, both with certain caching and computing capabilities. 
In this paper, we focus on the on-demand $360^\circ$ VR video streaming. As mentioned above, instead of transmitting the whole $360^{\circ}$ video, the MEC server only delivers the requested FOV  at each time. Key notations in this paper are summarized in Table~\ref{notation}.

\subsection{VR Task Model}
Denote with $\mathcal{N} \triangleq \{1,..., N\}$ the viewpoint space \cite{J}. 
The projection from 2D FOV to 3D FOV of each viewpoint $i\in \mathcal{N}$ is characterized
by a $3-$tuple $(D^I,D^O,w)$, where $D^I$ and $D^O$ are the data sizes (\emph{in bit}) of the 2D FOV and 3D FOV, respectively, and $w$ is the number of computation cycles required to process one bit input (\emph{in cycle/bit}). Denote with $\alpha \triangleq \frac{D^O}{D^I}$ the ratio of the size of 3D FOV to that of 2D FOV. Typically, $\alpha \geq 2$ in order to create a stereoscopic vision \cite{3D}. 
\subsection{Request Model}
The request stream at the mobile VR device conforms to the independent reference model (IRM) \cite{J} based on the following assumptions: i) the viewpoints that the mobile VR device requests are fixed to the set $\mathcal{N}$; ii) the probability of the  request for viewpoint $i\in \mathcal{N}$ at the mobile VR device \textcolor{black}{at each time}, denoted as $P_{i}$, is constant and independent of all the past requests, \textcolor{black}{satisfying} $\sum_{i=1}^N P_{i} = 1$. We consider uniform distribution, i.e., $P_i = \frac{1}{N}$ for each $i \in \mathcal{N}$.\footnote{The scenario with nonuniform data size and popularity distribution is considered in Section~\ref{heterogeneous}.}
 In addition,  in order to avoid dizziness and nausea, each request at the mobile VR device must be satisfied within the deadline of $\tau$ (\emph{in second}). 



\subsection{Caching and Computing Model}
First, consider the cache placement at the mobile VR device. We assume that the mobile VR device is equipped with a cache size $CD^I$ (\emph{in bit}), where $C$ is an integer, and is able to store both 2D and 3D FOVs of some viewpoints.
Denote with $c_i^I \in \{0,1\}$ the caching decision for 2D FOV of viewpoint $i$, where $c_i^I = 1$ means that the 2D FOV of viewpoint $i$ is cached at the mobile VR device and $c_i^I = 0$ otherwise. Denote with $c_i^O \in \{0,1\}$ the caching decision for 3D FOV of viewpoint $i$, where $c_i^O = 1$ means that the 3D FOV of viewpoint $i$ is cached at the mobile VR device and $c_i^O = 0$ otherwise. 
Under the cache size constraint of the mobile VR device, we have
\begin{equation}\label{CacheSize}
\sum_{i=1}^N  D^Ic_i^I + \alpha D^I c_i^O \leq CD^I.
\end{equation}
For the cache placement at the MEC server, we assume that both 2D and 3D FOVs of all the viewpoints 
are cached at the MEC server. This is reasonable due to the fact that the storage size at the MEC server is much larger than that of the mobile VR device.

Next,  consider the computing decision for the projection component at the mobile VR device. The mobile VR device is assumed to run at a \textcolor{black}{given} CPU-cycle frequency, denoted as $f_V$ (\emph{in cycle/s}), \textcolor{black}{and has an average energy constraint, denoted as $\bar{E}$ (\emph{in J})}. 
The energy consumed for computing one cycle with frequency $f_V$ at the mobile VR device is $k f_V^2$, where $k$ is a constant related to the hardware architecture and can indicate the power efficiency of CPU at the mobile VR device \cite{cycle}. 
Denote with $d_i \in \{0,1\}$ the computing decision for viewpoint $i$, where $d_i = 1$ indicates that the projection \textcolor{black}{from 2D FOV to 3D FOV is executed at the mobile VR device upon viewpoint request and} $d_i = 0$ otherwise.  
\textcolor{black}{Under the average energy consumption constraint of the mobile VR device,} we have
\begin{equation}\label{energy}
\frac{k f_V^2D^Iw}{N}\sum_{i=1}^N d_i \leq \bar{E}.
\end{equation}
 From (\ref{energy}), note that $ \frac{N\bar{E}}{k f_V^2D^Iw}$ corresponds to the maximum number of projections that can be computed at the mobile VR device, named as \textit{computing capability} of the mobile VR device, and is assumed to be an integer throughout this paper. 

Last, denote with $(\mathbf{c}^O, \mathbf{c}^I,\mathbf{d})$ the joint caching and computing decision at the mobile VR device, where $\mathbf{c}^O \triangleq (c_i^O)_{i \in \mathcal{N}}$ denotes the caching decision vector for 3D FOVs of all the viewpoints and $\mathbf{c}^I \triangleq (c_i^I)_{i\in \mathcal{N}}$ denotes the caching decision vector for 2D FOVs of all the viewpoints, satisfying the cache size constraint in (\ref{CacheSize}), and $\mathbf{d} \triangleq (d_i)_{i\in \mathcal{N}}$ denotes the computing decision vector, satisfying the local average energy consumption constraint in (\ref{energy}).

\subsection{Service Mechanism and Transmission Rate Requirement}

Based on the joint caching and computing decision, i.e., $(\textbf{c}^O,\textbf{c}^I,\textbf{d})$, we can see that request for viewpoint $i \in \mathcal{N}$ can be served via the following four routes, each of which yields a unique minimum transmission rate requirement, denoted as $R_i$ (\emph{in bit/s}).
\begin{itemize}
\item \textbf{Local 3D caching.} If $c_i^O=1$, 
the 3D FOV of viewpoint $i$ can be obtained from the local cache and without the need of the transmission and computing. In this way,  the required latency is negligible and the minimum required transmission rate is $R_i=0$.
\item \textbf{Local computing with local 2D caching.} If $c_i^O = 0$, $d_i = 1$ and $c_i^I=1$, the mobile VR device obtains the 2D FOV of viewpoint $i$ from the local cache and without the need of transmission, and then projects it to 3D FOV using its local CPU processor. Thus, the overall consumed latency is $\frac{D^Iw}{f_V}$ (\emph{in second}) and \textcolor{black}{the minimum required transmission rate is} $R_i = 0$. In this paper, we assume that $\frac{D^Iw}{f_V}\! <\! \tau$ for feasibility, i.e., computing the projection at the mobile VR device can be completed within the deadline.

\item \textbf{Local computing without local caching.} If $c_i^O = 0$, $d_i = 1$ and $c_i^I=0$, the mobile VR device downloads the 2D FOV of  viewpoint $i$ from the MEC server and then projects it to 3D FOV using its local CPU processor. Thus, the overall consumed latency is $\frac{D^I}{R_i} + \frac{D^Iw}{f_V}$ (\emph{in second}), where $\frac{D^I}{R_i}$ corresponds to the 2D FOV transmission latency over the wireless link and  $\frac{D^Iw}{f_V}$ corresponds to the computation latency at the mobile VR device. Under the latency constraint,  the minimum required transmission rate is $ R_i  = R_V \triangleq \frac{D^I}{\tau-\frac{D^Iw}{f_V}}.$

\item \textbf{MEC computing.} If $c_i^O = 0$ and $d_i=0$, the mobile VR device downloads the 3D FOV of viewpoint $i$ from the MEC server. Then, the overall consumed latency can be represented as $\frac{D^O}{R_i}$ (\emph{in second}). 
Under the latency constraint, the minimum required transmission rate is $ R_i = R_S \triangleq \frac{D^O}{\tau}.$

\end{itemize}

By combining all the above cases, for any given joint caching and computing decision $(\textbf{c}^O,\textbf{c}^I,\textbf{d})$,   the minimum average required transmission rate to deliver the requested 3D FOV under latency constraint, denoted as $\bar{R}$ (\emph{in bit/s}), is given by
\begin{equation}\label{averagerate111}
\bar{R} = \frac{1}{N} \sum_{i=1}^N \left( R_Vd_i(1-c_i^I) + R_S\left(1-d_i\right) \right)\left(1-c_i^O\right).
\end{equation}
Obviously, minimizing $\bar{R}$ is equivalent to minimizing the bandwidth for a given spectral efficiency.
\begin{table}[t]
\caption{Transmission Rates vs. Local Caching and Computing Costs}\label{tradeoff}
\newcommand{\tabincell}[2]{\begin{tabular}{@{}#1@{}}#2\end{tabular}}
\newcommand{\tl}[1]{\multicolumn{1}{l}{#1}} 
\renewcommand\arraystretch{1}
\begin{center}
\setlength{\tabcolsep}{1.5mm}{
\begin{tabular}{lccc}
\hline
Joint Decision & Rate  & Caching & Computing  \\
\hline
\tabincell{l}{Local 3D caching\\$c_i^O =1, c_i^I =0, d_i =0$}& $0$ & $\alpha D^I$ & $0$\\
\hline
\tabincell{l}{Local computing  with local 2D caching \\$c_i^O =0, c_i^I =1, d_i =1$}& $0$ & $D^I$ & $\frac{kD^Iwf_V^2}{N}$\\
\hline
\tabincell{l}{Local computing without local caching \\$c_i^O =0, c_i^I=0, d_i =1$}  & $\frac{R_V}{N}$ & $0$ & $\frac{kD^Iwf_V^2}{N}$ \\
\hline
\tabincell{l}{MEC computing \\ $c_i^O =0, c_i^I=0, d_i=0$} & $\frac{R_S}{N}$ & $0$ & $0$\\
\hline
\end{tabular}}
\end{center}
\label{a}
\end{table}%

\begin{remark}\label{gain}
As illustrated in Table~\ref{tradeoff}, for each viewpoint $i \in \mathcal{N}$, compared with local 3D caching, local computing with local 2D caching achieves the same rate gain and saves at least half of the cache size consumed by local 3D caching, but incurs additional energy consumption; compared with  local computing with local 2D caching, local computing without local caching saves cache cost, but incurs larger transmission rate requirement; compared with local computing without local caching, MEC computing saves local caching and computing cost, but relationship between its incurred transmission rate, i.e., $R_S$, and that incurred by local computing without local caching, i.e., $R_V$, depends on the local computing capability. Thus,  joint caching and computing design requires careful thinking.
\end{remark}
\section{Problem Formulation and Optimal Property Analysis}
In this section, we formulate the joint caching and computing optimization problem to minimize the average required transmission rate and analyze the optimal properties, based on which we obtain an equivalent problem.
\subsection{Problem Formulation}
\begin{Prob}[Joint Caching and Computing Optimization]\label{Prob1}
\begin{align}
& \min_{{\mathbf{c}^O,\mathbf{c}^I,\mathbf{d}}}\ \ \ \frac{1}{N} \sum_{i=1}^N \left( R_Vd_i(1-c_i^I) + R_S\left(1-d_i\right) \right)\left(1-c_i^O\right) \nonumber\\
&\ \ s.t. \ \ \ \ \ \ \ \ \ \ \ \ \ \ \ \ \ \ \ \ \ \ \  \sum_{i=1}^N  c_i^I + \alpha  c_i^O \leq C,\label{cc1} \\
&\ \ \ \ \ \ \ \ \ \ \ \ \ \ \ \ \ \ \ \ \ \ \ \ \ \ \ \ \  \sum_{i=1}^N d_i \leq  \frac{N\bar{E}}{k f_V^2D^Iw},\label{dd1}\\
&\ \ \ \ \ \ \ \ \ \ \ \ \ c_i^O\in \{0,1\},\ c_i^I \in \{0,1\},\ d_i \in \{0,1\},\ i \in \mathcal{N},\nonumber
\end{align}
\end{Prob}
\noindent
where (\ref{cc1}) and (\ref{dd1}) correspond to the cache size constraint in (\ref{CacheSize}) and average energy consumption constraint in (\ref{energy}), respectively. Denote with $R^*$ the optimal objective value of Problem \ref{Prob1} and $(\mathbf{c}^{O^*}, \mathbf{c}^{I^*}, \mathbf{d}^*$) the optimal joint caching and computing decision.



\subsection{Optimal Properties and Equivalent Formulation}
In this subsection, we analyze the optimal properties of the joint caching and computing policy, based on which we obtain an equivalent optimization. Denote with $c^O \triangleq \sum_{i=1}^N c_i^O$, $c^I \triangleq \sum_{i=1}^N c_i^I$ and $d \triangleq \sum_{i=1}^N d_i$ the number of locally cached 3D FOVs, that of locally cached 2D FOVs and that of locally computed projections, respectively. From (\ref{cc1}) and (\ref{dd1}), we have $c^O \in \left\{0,1,\cdots,\frac{C}{\alpha}\right\}$, $c^I \in \left\{0,1,\cdots,C-\alpha c^O\right\}$ and $d\in \left\{0,1,\cdots,\frac{N\bar{E}}{k f_V^2D^Iw}\right\}$, respectively. Considering that the projection tuple ($D^I,D^O,w, P_i, \tau$) of each viewpoint $i\in \mathcal{N}$ is the same, for any given $c^O$, we 
can let
\begin{equation}\label{num3D}
c_i^O =
\begin{cases}
1& \text{$i = 1, \cdots, c^O$,}\\
0& \text{otherwise,}
\end{cases}
\end{equation}
without loss of optimality.

We first obtain the optimality property between local 2D and 3D FOV caching.
\begin{property}\label{struc1}
For any $i\in \mathcal{N}$ such that $c_i^O=1$, we have $c_i^I=0$.
\end{property}
This property indicates that if  3D FOV of viewpoint $i$ is already cached at the mobile VR device, there is no need to cache the 2D FOV, since the request for viewpoint $i$ can be directly served from  local cache. 
 \begin{property}\label{strucC}
 For any given $c^O$, we have $c^I = C-\alpha c^O$.
 \end{property}

Property~\ref{strucC} can be obtained by observing that the equality holds in the cache size constraint (\ref{cc1}) for minimizing the required transmission rate. Based on Property~\ref{struc1} and Property~\ref{strucC}, when $\textbf{c}^O$ is given by (\ref{num3D}), we can let
\begin{equation}\label{num2D}
c_i^I =
\begin{cases}
0& \text{$i = 1, \cdots, c^O$,}\\
1& \text{$i = c^O+1,\cdots,c^O+c^I$,}\\
0& \text{otherwise},
\end{cases}
\end{equation}
where $c^I = C-\alpha c^O$. 

We next analyze the optimality between local caching and local computing as follows.

 \begin{property}\label{optimalrelation2}
 For any viewpoint $i\!\in\! \mathcal{N}$, we have $c_i^O+d_i \leq 1$ and $c_i^I \leq d_i$.
 \end{property}
 Property~\ref{optimalrelation2} can be obtained by contradiction. First, suppose that $c_i^O+d_i > 1$. Then, when $c_i^O=1$, we have $d_i=1$. However, when $c_i^O=1$, by setting $d_i$ from $1$ into $0$, $\bar{R}$ does not change and computing cost is saved. Thus, $c_i^O+d_i > 1$ is not optimal. Secondly, suppose that $c_i^I > d_i$. Then, when $d_i=0$, we have $c_i^I=1$. However, when $d_i=0$, by setting $c_i^I$ from $1$ into $0$, based on (\ref{averagerate111}), $\bar{R}$ does not change and caching cost is saved. Thus, $c_i^I>d_i$ is not optimal. 

Property~\ref{optimalrelation2} indicates that if 3D FOV of viewpoint $i$ is already cached at the mobile VR device,  there is no gain from local computing, since the request for viewpoint $i$ can be directly served from the local cache. Similarly, if 2D FOV is already cached at the mobile VR device, it would be a waste of caching resource if the locally cached 2D FOV is not utilized to compute the projection component at the mobile VR device. 

Based on Property~\ref{optimalrelation2}, when $\textbf{c}^O$ and $\textbf{c}^I$ are given by (\ref{num3D}) and (\ref{num2D}), for any given $d$, we can let

\begin{equation}\label{numOffload}
d_i =
\begin{cases}
0& \text{$i = 1, \cdots, c^O$,}\\
1& \text{$i = c^O+1,\cdots,c^O+d$,}\\
0& \text{otherwise}.
\end{cases}
\end{equation}

Finally, for ease of structural property analysis, by rewriting $\textbf{c}^O$, $\textbf{c}^I$ and $\textbf{d}$ as (\ref{num3D}), (\ref{num2D}) and (\ref{numOffload}), Problem \ref{Prob1} is equivalent to Problem~\ref{Prob2}.

\begin{Prob}[Equivalent Joint Policy Optimization]\label{Prob2}
\begin{align}
& \min_{c^O,c^I,d} R_S - \frac{R_S}{N} c^O - \frac{R_S}{N} \min \left\{c^I, d\right\} \nonumber \\
&\hspace{40mm}-  \frac{R_S-R_V}{N}\left(d-\min \left\{c^I,d\right\}\right) \nonumber\\
&\  s.t.\ \ \ \ \ \ \  \ \ \ \ \ \ \ \ \ \ \ \ \ \ \ \ c^I \in \left\{0,1,\cdots,C\right\}, \\
& \ \ \ \ \ \ \  \ \ \ \ \ \ \ \ \ \ \ \ \ \ \ \ \ \ \ \ \ c^O = \frac{C - c^I}{\alpha}, \\
& \ \ \ \ \ \ \  \ \ \ \ \ \ \ \ \ \ \ \ \ \ \ \ \  d \in \left\{0,1,\cdots,\frac{N\bar{E}}{k f_V^2D^Iw}\right\}.
\end{align}
\end{Prob}
\noindent
Denote with $\left(c^{O^*},c^{I^*},d^*\right)$ the optimal solution to Problem~\ref{Prob2}. Based on (\ref{num3D}), (\ref{num2D}) and (\ref{numOffload}), we can obtain the corresponding optimal joint policy, i.e., $\left(\mathbf{c}^{O^*}, \mathbf{c}^{I^*}, \mathbf{d}^*\right)$.

From the objective function of Problem \ref{Prob2}, we note that the first term, i.e., $R_S$, corresponds to the average transmission rate required via MEC computing and without local caching or local computing. 
 The second term, i.e., $\frac{R_S}{N} c^O$, corresponds to the local 3D caching gain, which increases with the number of locally cached 3D FOVs, i.e., $c^O$. The third term, i.e., $\frac{R_S}{N} \min \left\{c^I,d\right\}$, corresponds to the local  computing gain with local 2D caching, which increases with the minimum of the number of locally cached 2D FOVs, i.e., $c^I$, and that of locally computed projections, i.e., $d$. The last term, i.e., $\frac{R_S-R_V}{N}\left(d-\min \left\{c^I,d\right\}\right)$, corresponds to the local computing gain without local caching, which depends on the difference between $R_S$ and $R_V$. Note that if $f_V < F \triangleq \frac{D^Ow}{(\alpha-1)\tau}$, $R_S < R_V$ and the local computing gain without local caching is negative. Thus, we name $f_V < F$ as \emph{local computing limited region}. Otherwise, the local computing gain without local caching is positive and we name $f_V \geq F $ as \emph{MEC computing limited region}. \textcolor{black}{In summary, the total number of viewpoint requests that can be served locally is $c^O+d$. For the interest of joint caching and computing design, we assume that $\frac{C}{\alpha} + \frac{N\bar{E}}{kf_V^2D^Iw} \leq N$.}

\section{Optimal Policy and Tradeoff Analysis}

In this section, we obtain the optimal joint caching and computing policy and the minimum transmission rate, yielding the fundamental relationship between communications, caching and computing, defined as \textbf{3C tradeoff}, in the local computing limited region, i.e., $f_V<F$, and MEC computing limited region, i.e., $f_V \geq F$, respectively.

\subsection{Local Computing Limited Region}
\begin{theorem}[Optimal joint policy and 3C tradeoff when $f_V \!< F$]\label{ror1}
The optimal joint policy $\left(c^{O^*}\!,c^{I^*}\!, d^*\right)$ is given as
\begin{equation}\label{3D1}
c^{O^*} = \frac{C - c^{I^*}}{\alpha},
\end{equation}
\begin{equation}\label{2D1}
\ \ \ \ \ \ \ \ \ \ \ \ \ \ c^{I^*} = \min \left\{C, \frac{N\bar{E}}{k f_V^2D^Iw}\right\},
\end{equation}
\begin{equation}\label{off1}
\ \ \ \ \ \ \ \ \ \ \ \ \ d^* = \min \left\{C, \frac{N\bar{E}}{k f_V^2D^Iw}\right\},
\end{equation}
and the minimum transmission rate $R^*$ is given as
\begin{equation}\label{rate1}
R^* = R_S - \frac{R_S}{N}\left(\frac{C}{\alpha}+\left(1-\frac{1}{\alpha}\right)\min \left \{C, \frac{N\bar{E}}{k f_V^2D^Iw}\right\}\right).
\end{equation}
\end{theorem}
\begin{proof}
Proof can be seen in Appendix~A.
\end{proof}
\begin{remark}[Tradeoff analysis when $f_V < F$]

When $f_V < F$, $R_S < R_V$ and from the objective function of Problem~\ref{Prob2}, we can see that the performance gain comes from local 3D caching and local computing with local 2D caching, but the local computing gain without local caching, i.e., $\frac{R_S-R_V}{N}\left(d\!-\!\min\left\{c^I,d\right\}\right)$, is negative. Specifically, from the optimal computing policy in (\ref{off1}), the optimal number of locally computed projections $d^*$ is jointly limited by the local caching and computing capabilities, since local computing without local 2D caching cannot bring rate gain. From the optimal caching policy in  (\ref{3D1}) and (\ref{2D1}), we can see that local computing with local 2D caching is chosen first and then local 3D caching is chosen if there still remains underutilized storage size, which indicates that caching resource is exploited more efficiently joint with computing resource, and vice versa. This is because the caching cost for local 2D caching of each FOV, i.e., $D^I$, is smaller than that for local 3D caching of each FOV, i.e., $D^O=\alpha D^I$.

\end{remark}

\begin{figure}[t]
\begin{center}
 \includegraphics[width = 8.5cm]{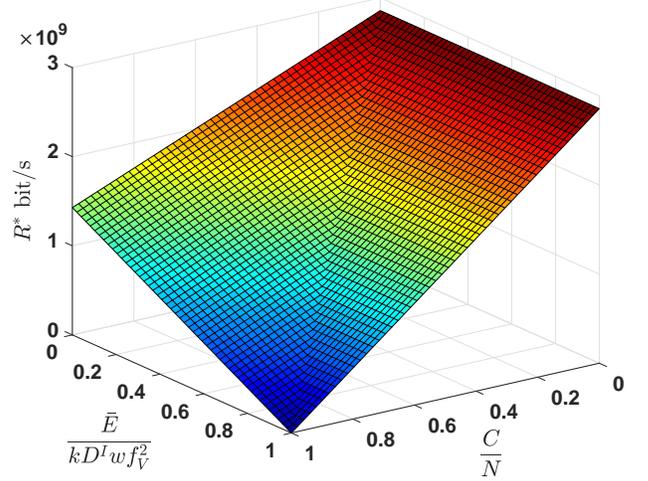}
\end{center}
 \caption{\small{3C tradeoff under latency constraint when $f_V < F$, where \textbf{any $(R^*,\ \frac{C}{N},\ \frac{\bar{E}}{kf_V^2D^Iw})$ point in this 3D figure achieves $\tau=20$ ms}. $D^I = 25$ M bits, $w = 10$ cycles/bit, $D^O = 50$ M bits, $N = 6 \times 10^4$, $k = 10^{-27}$, $f_V = 0.7*F$. }
 }
\label{3C1}
\end{figure}
As illustrated in Fig.~\ref{3C1}, we can see that $R^*$ first decreases with  the local computing capability and then remains unchanged when the local computing capability is larger than local caching capability. Thus, the caching capability facilitates the utilization of the local computing capability when $f_V < F$.

Based on Theorem~\ref{ror1}, we analyze the impacts of cache size $C$ and local computation frequency $f_V$ on the average transmission rate $R^*$ by plotting numerical results.


Fig.~\ref{case1} (a) illustrates the impact of the cache size $C$ on the optimal rate $R^*$ when $f_V<F$.  We can see that the decreasing rate of $R^*$ w.r.t. $C$ depends on the relationship between the caching $C$ and computing $\frac{N\bar{E}}{k f_V^2D^Iw}$ capabilities of the mobile VR device. When $C \leq \frac{N\bar{E}}{k f_V^2D^Iw}$, $R^*$ decreases with $C$ at the rate of $\frac{R_S}{N}$ since the caching gain comes from local computing with local 2D caching. Otherwise, $R^*$ decreases with $C$ at the rate of $\frac{R_S}{\alpha N}$ since the caching gain comes from local 3D caching.

Fig.~\ref{case1} (b) illustrates the impact of the computation frequency $f_V$ on the optimal rate $R^*$ when $f_V<F$. We can see that $R^*$ increases with $f_V$ and $k$. This is because increasing $f_V$ or $k$ decreases the number of projections that can be computed at the mobile VR device. Since increasing $k$ corresponds to decreasing the power efficiency, we learn that improving the power efficiency of the mobile VR device can help facilitate utilizing the local computing resource and thereby reduce the transmission rate requirement. In addition, we observe that $R^*$ decreases with $\bar{E}$ since the increase of $\bar{E}$ increases the number of projections that can be computed at the mobile VR device.
\begin{figure}[t]
\begin{center}
 \subfigure[Cache size when $\frac{\bar{E}}{kf_V^2D^Iw} = 30\%$.]
  {\resizebox{7.5cm}{!}{\includegraphics{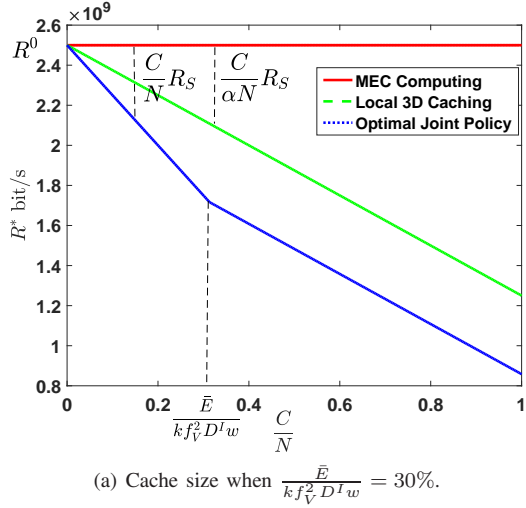}}}\quad\quad
 \subfigure[Computation frequency when $\frac{C}{N} = 30\%$.]
 {\resizebox{7.5cm}{!}{\includegraphics{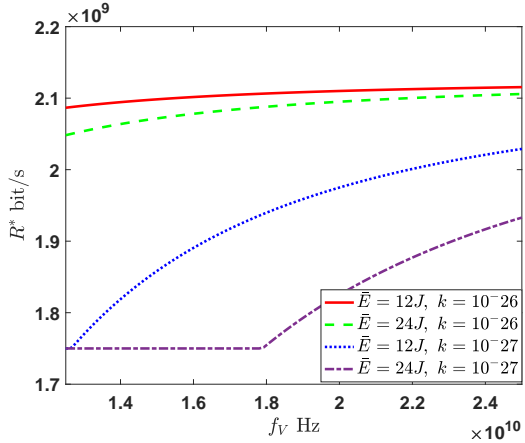}}}\quad\quad
\end{center}
   \caption{\small{Local cache size and computation frequency when $f_V < F$. Parameters are the same as those in Fig.~\ref{3C1}}.}
\label{case1}
\end{figure}

\subsection{MEC Computing Limited Region}
\begin{theorem}[Optimal joint policy and 3C tradeoff when $F\! \leq\! f_V$]\label{ror2}
The optimal joint policy, i.e., $(c^{O^*}\!, c^{I^*}\!,d^*)$,  is given as
\begin{equation}\label{3D2}
c^{O^*} = \frac{C - c^{I^*}}{\alpha},
\end{equation}
\begin{equation}\label{2D2}
\ \ \ \ \ \ \ \ \ \ \ \ \ \ \ \ \ \ \ \ \ \ \ \ \ \ c^{I^*} = \min \left\{C, \frac{N\bar{E}}{k f_V^2D^Iw}\right\},
\end{equation}
\begin{equation}\label{off2}
 d^* = \frac{N\bar{E}}{k f_V^2D^Iw},
\end{equation}
and the minimum transmission rate $R^*$ is given as
\begin{align}\label{rate2}
&R^* = R_S - \frac{R_S}{N}\left(\frac{C}{\alpha} + \left(1-\frac{1}{\alpha}\right) \min \left\{C, \frac{N\bar{E}}{k f_V^2D^Iw}\right\}\right) \nonumber\\
&\hspace{4mm}- \frac{R_S-R_V}{N}\left(\frac{N\bar{E}}{k f_V^2D^Iw}-\min \left\{C, \frac{N\bar{E}}{k f_V^2D^Iw}\right\}\right).
\end{align}
\end{theorem}
\begin{proof}
Proof can be seen in Appendix~B.
\end{proof}
\begin{remark}[Tradeoff analysis when $F \leq f_V$]
When $F \leq f_V$, $R_S \geq R_V$ and we can see that the performance gain comes from local 3D caching, local computing with local 2D caching as well as  local computing without local caching. 
Specifically, from the optimal computing policy in (\ref{off2}), the optimal number of locally computed projections $d^*$ is only limited by the computing capability, i.e., $\frac{N\bar{E}}{k f_V^2D^Iw}$, since local computing without local 2D caching can also bring performance gain. From the optimal caching policy in  (\ref{3D2}) and (\ref{2D2}),  the local computing with local 2D caching is chosen first and then the 3D caching is chosen if there still remains underutilized storage size. The reasons lie in the following two aspects. First, the caching cost for local 2D caching of each FOV, i.e., $D^I$, is smaller than that for local 3D caching of each FOV, i.e., $D^O=\alpha D^I$. Secondly, when $ F\leq f_V$, the gain from local computing without caching, i.e., $R_S-R_V$, is not large enough to compensate for the caching resource waste for 3D caching compared with 2D caching.
\end{remark}
\begin{figure}[t]
\begin{center}
 \includegraphics[width = 8cm]{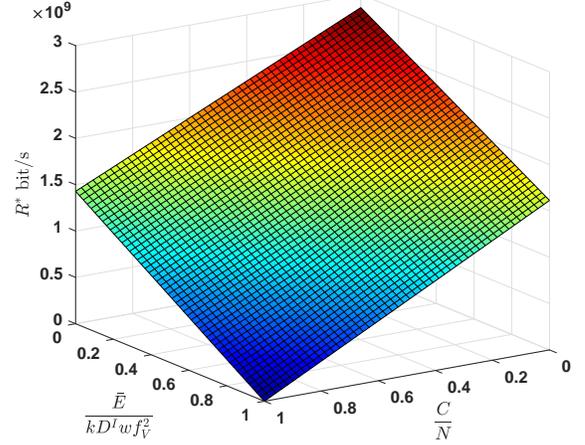}
\end{center}
 \caption{\small{3C Tradeoff when $F \leq f_V$, where \textbf{any $(R^*,\ \frac{C}{N},\ \frac{\bar{E}}{kf_V^2D^Iw})$ point in this 3D figure can achieve $\tau=20$ ms}. $f_V = 1.2*F$. Other parameters are the same as those in Fig.~\ref{3C1}.}
 }
\label{3C2}
\end{figure}
As illustrated in Fig.~\ref{3C2}, we can see that $R^*$ monotonically decreases with  the local computing capability and the decreasing rate increases with the cache capability ($43 \%$ when $\frac{C}{N}=0$ and $100 \%$ when $\frac{C}{N}=1$). It demonstrates that the caching capability facilitates the utilization of the local computing capability. On the other hand, $R^*$ monotonically decreases with the local caching capability and the decreasing rate increases with the local computing capability ($50 \%$ when $\frac{\bar{E}}{kf_V^2D^Iw}=0$ and $100\%$ when $\frac{\bar{E}}{kf_V^2D^Iw}=1$). It also demonstrates that the local computing capability facilitates the utilization of the caching capability.

Based on Theorem~\ref{ror2}, we analyze the impacts of cache size, i.e., $C$, and local computation frequency, i.e., $f_V$, on $R^*$ via plotting numerical results.


\begin{figure}[t]
\begin{center}
 \subfigure[Cache size when $\frac{\bar{E}}{kf_V^2D^Iw} = 30\%$.]
  {\resizebox{7.5cm}{!}{\includegraphics{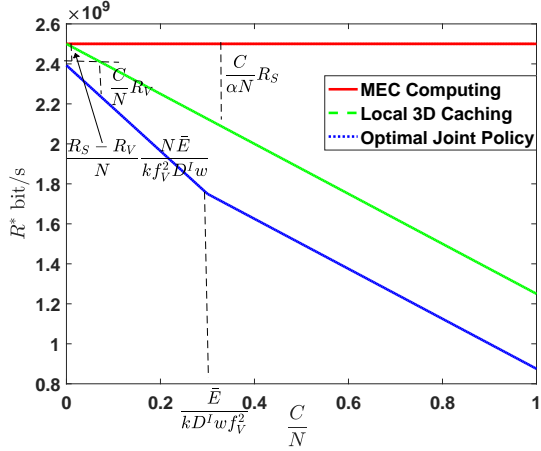}}}\quad\quad
 \subfigure[Computation frequency when $\frac{C}{N} = 30 \%$.]
 {\resizebox{7.5cm}{!}{\includegraphics{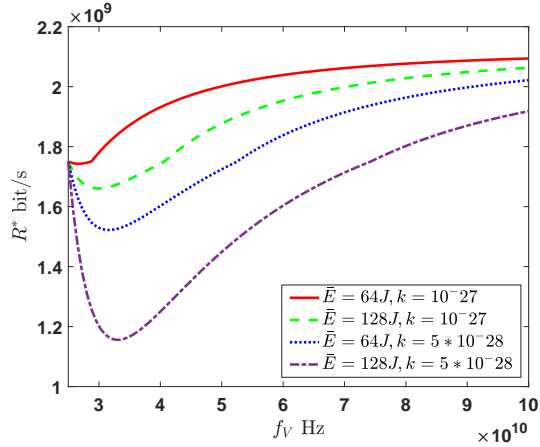}}}\quad\quad
\end{center}
   \caption{\small{Local cache size and computation frequency when $F \leq f_V$. Parameters are the same as those in Fig.~\ref{3C2}}.}
\label{case2}
\end{figure}
Fig.~\ref{case2} (a) illustrates the impact of $C$ on the optimal rate $R^*$ when $F \leq f_V$. We can see that it exhibits similar structure to that when $f_V < F$. The decreasing rate of $R^*$ w.r.t. $C$ depends on the relationship between the caching $C$ and computing $\frac{N\bar{E}}{k f_V^2D^Iw}$ capabilities of the mobile VR device, i.e., $\frac{R_V}{N}$ when $C \leq \frac{N\bar{E}}{k f_V^2D^Iw}$ and $\frac{R_S}{\alpha N}$ otherwise.
Note that when $C=0$, there still exists performance gain, i.e., $\frac{R_S-R_V}{N} \frac{N\bar{E}}{k f_V^2D^Iw}$, due to the  local computing without local 2D caching.

Fig.~\ref{case2} (b) illustrates the impact of $f_V$ on the optimal rate $R^*$ when $F \leq f_V$.  We can see that $R^*$ first decreases and then increases with $f_V$. This is mainly due to the fact that when $f_V$ is relatively small, increasing $f_V$ alleviates the transmission rate requirement by reducing the local computation latency, while when $f_V$ is relatively large, increasing $f_V$ decreases the number of projections that can be computed at the mobile VR device. In addition, we can see that $R^*$ decreases with $\bar{E}$ and increases with $k$, which demonstrates again that improving power efficiency of the mobile VR device can help facilitate the local computing gain. From the first-order derivative of $R^*$ w.r.t. $f_V$, we obtain the following remark.
\begin{remark}\label{cof2}
 When $C = 0$ and $F \leq f_V$, $f_V^*$ minimizing $R^*$ is given by
\vspace{-2mm}
\begin{equation}\label{f_V}
 f_V^* = \Big(1-\frac{D^I}{4R_S\tau} \Big)F + \sqrt{\Big(1-\frac{D^I}{4R_S\tau}\Big)^2F^2-\frac{D^Iw}{\tau}F}.
\end{equation}
\end{remark}
Equation (\ref{f_V}) indicates that the optimal computation frequency is independent of the energy $\bar{E}$ and power efficiency $k$ of the mobile device, and depends on the projection parameters $(D^I,D^O,w,\tau)$ only.

\section{Problem Formulation in Heterogeneous Scenario}\label{heterogeneous}

In this section, we consider a heterogeneous scenario, where the parameters of each viewpoint $i \in \mathcal{N}$, generalized as $(D_i^I,D_i^O,w_i,\tau_i, P_i)$, are different from each other. 
Similar to Problem~\ref{Prob2}, 
the optimization problem is formulated as below.

\begin{Prob}[Joint Policy Optimization in Heterogeneous Scenario]\label{generalProb}
\begin{align}
\ \ \  \min_{{\mathbf{c}^I,\mathbf{c}^O,\mathbf{d}}}&\ \ \  \sum_{i=1}^NP_i\left(R_i^S\left(1-d_i\right) + R_i^Vd_i\left(1-c_i^I\right)\right)\left(1-c_i^O\right) \nonumber \\
s.t. \ \  &\ \ \ \ \ \ \ \ \ \ \ \ \ \sum_{i=1}^N P_ikf_V^2D_i^Iw_id_i \leq \bar{E}, \\
  \ \ \ & \ \ \ \ \ \ \ \ \ \ \ \ \ \sum_{i=1}^N D_i^Ic_i^I+\alpha D_i^Ic_i^O \leq C' ,\\\nonumber
\ \ \   & \ \ \ \ \ \ c_i^O\in \{0,1\},\ c_i^I \in \{0,1\},\ d_i \in \{0,1\},\ i \in \mathcal{N},\nonumber
\end{align}
\end{Prob}
\noindent
where $R_i^S \triangleq \frac{D_i^O}{\tau_i}$ (\emph{in bit/s}) and $R_i^V \triangleq \frac{D_i^I}{\tau_i-\frac{D_i^Iw_i}{f_V}}$ (\emph{in bit/s}) denote the minimally required transmission rates to satisfy the latency constraint when the projection of viewpoint $i \in \mathcal{N}$ is computed at the MEC server and at the mobile VR device, respectively. The objective function is obtained directly via generalizing (\ref{averagerate111}). 
$C'$ (\emph{in bit}) denotes the cache size at the mobile VR device.

\begin{table}[t]
\caption{Gains vs. Caching and Computing Costs in Heterogeneous Scenario}\label{tradeoff2}
\newcommand{\tabincell}[2]{\begin{tabular}{@{}#1@{}}#2\end{tabular}}
\newcommand{\tl}[1]{\multicolumn{1}{l}{#1}} 
\renewcommand\arraystretch{1}
\begin{center}
\setlength{\tabcolsep}{0.8mm}{
\begin{tabular}{clccc} 
\hline
Route & Joint Decision & Rate Gain & Caching & Computing \\
\hline
Route~1&\tabincell{l}{Local 3D caching\\$c_i^O=1, c_i^I=0, d_i=0$} & $P_iR_i^S$ & $\alpha D_i^I$ & $0$\\
\hline
Route~2&\tabincell{l}{Local computing \\ with local 2D caching\\$c_i^O=0, c_i^I=1, d_i=1$}& $P_iR_i^S$ & $D_i^I$ & $P_ikD_i^Iw_if_V^2$\\
\hline
Route~3&\tabincell{l}{Local computing\\ without local caching\\$c_i^O=0, c_i^I=0,d_i=1$} & $P_i\left(R_i^S\!-R_i^V\right)$ & $0$ & $P_ikD_i^Iw_if_V^2$ \\
\hline
Route~4 & \tabincell{l}{MEC computing \\ $c_i^O=0,c_i^I=0, d_i=0$}& $0$ & $0$ & $0$\\
\hline
\end{tabular}}
\end{center}
\label{a}
\end{table}%

%
For each viewpoint $i\in \mathcal{N}$,  we list the transmission rate gain compared with the MEC computing, local caching and computing costs in Table~\ref{tradeoff2} of each service route, which is obtained via directly generalizing Table~\ref{tradeoff}. In the following, we will show that Problem~\ref{generalProb} is NP-hard in strong sense and transform Problem~\ref{generalProb} into an equivalent IQP, which can be solved via CCCP efficiently. 

%
%

\subsection{Computational Intractability}
To show that Problem~\ref{generalProb} is NP-hard in strong sense, we transform Problem~\ref{generalProb} into a multiple choice multiple dimensional knapsack problem (MMKP) equivalently. For each viewpoint $i\in \mathcal{N}$ and service route $j \in \{1,2,3,4\}$, introduce variable $x_{i,j} \in \{0,1\}$ where $x_{i,j}=1$ indicates that the request for viewpoint $i$ is served via the $j$-th route and $x_{i,j}=0$ otherwise.  Based on Table~\ref{a}, $\left(\mathbf{c}^{O^*}, \mathbf{c}^{I^*}, \mathbf{d}^*\right)$ can be obtained from $(x_{i,j})_{i\in \mathcal{N}, j \in \{1,2,3,4\}}$, and without loss of equivalence, Problem~\ref{generalProb} can be rewritten as Problem~\ref{mmkp}.
\begin{Prob}[Equivalent Joint Policy Optimization]\label{mmkp}
\begin{align}
& \max_{(x_{i,j})_{i\in \mathcal{N},j\in \{1,2,3,4\}}} \ \ \ \ \ \ \  \ \ \  \sum_{i = 1}^N \sum_{j=1}^4 v_{i,j}x_{i,j}\nonumber \\
&\  \ \ \ \ \ \ s.t. \ \ \ \ \ \ \  \ \ \ \ \ \ \ \ \sum_{i=1}^N\sum_{j=1}^4 w_{i,j}^1x_{i,j} \leq C', \label{cache1}\\
&\ \ \ \ \ \ \  \ \ \ \ \ \ \ \ \ \ \ \ \ \ \ \ \ \ \ \sum_{i=1}^N\sum_{j=1}^4 w_{i,j}^2x_{i,j} \leq \bar{E},\label{energy2}\\
&\ \ \ \ \ \ \  \ \ \ \ \ \ \ \ \ \ \ \ \ \ \ \ \ \ \ \sum_{j=1}^4 x_{i,j} = 1, \ i \in \mathcal{N},\label{sum}\\
&\ \ \ \ \ \ \  \ \ \ \ \ \ \ \ \ \ \ x_{i,j} \in \{0,1\}, \ i \in \mathcal{N},\ j \in \{1,2,3,4\}\label{binary},
\end{align}
\end{Prob}
where
\begin{equation}\label{value}
\ \ \ \ \ \ \ \ \ \ \ v_{i,j} \triangleq
\begin{cases}
P_iR_i^S& \ \ \ \ \ \text{$j=1,2$,}\\
P_i(R_i^S-R_i^V)& \ \ \ \ \ \text{$j=3$,}\\
0 & \ \ \ \ \ \text{$j=4$,}
\end{cases}
\end{equation}
denotes the profit value for the choice of $j$ for viewpoint $i$,
\begin{equation}\label{value}
\ \ \ \ \ \ \ \ \ \ w_{i,j}^1 \triangleq
\begin{cases}
\alpha D_i^I& \ \ \ \ \ \ \ \ \ \ \ \ \ \ \ \ \text{$j=1$,}\\
D_i^I & \ \ \ \ \ \ \ \ \ \ \ \ \ \ \ \ \text{$j=2$,}\\
0 & \ \ \ \ \ \ \ \ \ \ \ \ \ \ \ \ \text{$j=3,4$,}
\end{cases}
\end{equation}
denotes the caching cost for the choice of $j$ for viewpoint $i$, and
\begin{equation}\label{value}
\ \ \ \ \ \ \ \  \ w_{i,j}^2 \triangleq
\begin{cases}
P_i kD_i^Iw_if_V^2& \ \ \ \ \ \ \  \text{$j=2,3$,}\\
0 & \ \ \ \ \ \ \ \text{$j=1,4$,}\\
\end{cases}
\end{equation}
denotes the energy cost for the choice of $j$ for viewpoint $i$.

We can see that Problem~\ref{mmkp} corresponds to a 4-choice 2-dimensional  knapsack problem. Since MMKP is NP-hard in strong sense \cite{mmkp}, we conclude that Problem~\ref{generalProb} is NP-hard in strong sense.

\subsection{Equivalent IQP and CCCP}
In the following, we transform Problem~\ref{mmkp} into an equivalent linearly constrained IQP and solve it using CCCP. First, without loss  of equivalence, (\ref{binary}) can be rewritten as
\begin{equation}\label{aa}
x_{i,j} \in [0,1], \ i \in \mathcal{N},\ j \in \{1,2,3,4\},
\end{equation}
\begin{equation}\label{b}
\sum_{i=1}^N \sum_{j=1}^4 x_{i,j}(1-x_{i,j}) \leq 0.
\end{equation}

Then, by substituting (\ref{binary}) with (\ref{aa}) and (\ref{b}), we transform Problem~\ref{mmkp} into Problem~\ref{equimmkp} equivalently.
\begin{Prob}[Equivalent Joint Policy Optimization]\label{equimmkp}
\begin{align}
& \min_{(x_{i,j})_{i\in \mathcal{N},j\in \{1,2,3,4\}}} \ \ \ \ \ \ \  \sum_{i = 1}^N \sum_{j=1}^4 -v_{i,j}x_{i,j} \nonumber \\
& \ \ \ \ \ \ \ \ \ s.t.  \ \ \ \ \ \ \  \ \   (\ref{cache1}), (\ref{energy2}), (\ref{sum}),(\ref{aa}),(\ref{b}). \nonumber
\end{align}
\end{Prob}

Note that Problem~\ref{equimmkp} is a continuous optimization problem, the computation complexity of which is much less compared with that of solving Problem~\ref{mmkp} directly. However, considering \!$\sum_{i=1}^N \!\sum_{j=1}^4\! x_{i,j}(1-x_{i,j})$ in (\ref{b}) is a concave function, (\ref{b}) is not a convex constraint and thus obtaining an efficient algorithm for solving Problem~\ref{equimmkp} is still very challenging.

Next, to facilitate the solution, we transform Problem~\ref{equimmkp} into Problem~\ref{penalized} by penalizing the concave constraint in (\ref{b}) to the objective function.

\begin{Prob}[Penalized Joint Policy Optimization]\label{penalized}
\begin{align}
& \min_{(x_{i,j})_{i\in \mathcal{N}, j\in \{1,2,3,4\}}}  \sum_{i = 1}^N \sum_{j=1}^4 -v_{i,j}x_{i,j}-\mu \sum_{i=1}^N\sum_{j=1}^4x_{i,j}(x_{i,j}-1)\nonumber \\
&\ \ \ \ \ \ \ \ \ s.t.  \ \ \ \ \ \ \  \ \ \ \ \ \ \ \ (\ref{cache1}), (\ref{energy2}), (\ref{sum}),(\ref{aa}), \nonumber
\end{align}
with the penalty parameter $\mu>0$. Denote with $\bar{R}(\mu)$ the optimal objective value.
\end{Prob}

Note that the objective function of Problem~\ref{penalized} is a difference of a linear function and a quadratic convex function, and the constraints of Problem~\ref{penalized} are linear. From \cite{localoptima}, Problem~\ref{penalized} is an IQP, a special case of general difference of convex problem, and the local optima of Problem~\ref{penalized} can be obtained in finite steps via DC algorithms (DCA). In addition, since the second term of the objective function of Problem~\ref{penalized} is differentiable, DCA exactly reduces to CCCP \cite{dcacccp}, as shown in Algorithm~\ref{dca}. CCCP involves iteratively solving a sequence of convex problems, each of which is obtained via linearizing the second term of the objective function of IQP. Specifically,  at each iteration $t$,  we approximate $\sum_{i=1}^N\sum_{j=1}^4x_{i,j}(x_{i,j}-1)$ with $\sum_{i=1}^N\sum_{j=1}^4x_{i,j}^{(t)}(x_{i,j}^{(t)}-1)+\sum_{i=1}^N\sum_{j=1}^4(2x_{i,j}^{(t)}-1)(x_{i,j}-x_{i,j}^{(t)})$. Thus, as for our problem, CCCP involves iteratively solving a sequence of linear problems, as shown in Algorithm~\ref{dca}.
\begin{algorithm}[H]
 \caption{CCCP for Solving Problem~\ref{penalized}}
\label{dca}
\small{\begin{algorithmic}[1]
\STATE \textbf{Initialization}. Find an initial feasible point $\textbf{x}^{(0)}$  of Problem~\ref{penalized} and set $t=0$.
\STATE \textbf{Repeat}
\STATE  Set $\textbf{x}^{(t+1)}$ to be an optimal solution to the following convex problem:
\begin{align}
& \min_{\textbf{x}}\ \ \  \ \ \ \ \ \ \  G(\textbf{x})-\mu \check{H}(\textbf{x}; \textbf{x}^{(t)})\nonumber\\
&\ s.t. ~\ \ \ \ \ \ \ \ \ (\ref{cache1}), (\ref{energy2}), (\ref{sum}),(\ref{aa}), \nonumber
\end{align}
where $G(\textbf{x}) \triangleq \sum_{i = 1}^N \sum_{j=1}^4 -v_{i,j}x_{i,j}$ and $\check{H}(\textbf{x}; \textbf{x}^{(t)})\triangleq \sum_{i=1}^N\sum_{j=1}^4x_{i,j}^{(t)}(x_{i,j}^{(t)}-1)+\sum_{i=1}^N\sum_{j=1}^4(2x_{i,j}^{(t)}-1)(x_{i,j}-x_{i,j}^{(t)})$.
\STATE Set $t\leftarrow t+1$.
\STATE \textbf{until} $\left[G\left(\textbf{x}^{(t-1)}\right)-\mu \check{H} \left(\textbf{x}^{(t-1)}; \textbf{x}^{(t-2)}\right)\right]- \left[G \left(\textbf{x}^{(t)}\right) - \mu \check{H} \left(\textbf{x}^{(t)}; \textbf{x}^{(t-1)}\right)\right] \leq \delta$.
\end{algorithmic}}
\end{algorithm}

Last, based on Theorem~1 in \cite{exactpenalty}, we show the equivalence between Problem~\ref{equimmkp} and Problem~\ref{penalized} in the following lemma.

\begin{lemma}[Exact Penalty]\label{exact} For all $\mu>\mu_0$ where
\begin{equation}
\mu_0 \triangleq \frac{\sum_{i=1}^N\sum_{j=1}^4 -v_{i,j}x_{i,j}^0-\bar{R}(0)}{\max_{\textbf{x}} \left\{\sum_{i=1}^N\sum_{j=1}^4x_{i,j}(x_{i,j}-1):(\ref{cache1}), (\ref{energy2}), (\ref{sum}),(\ref{aa})\right\} },
 \end{equation}
with any $(x_{i,j}^0)_{i\in \mathcal{N}, j \in \{1,2,3,4\}}$ satisfying (\ref{cache1}), (\ref{energy2}), (\ref{sum}) and (\ref{aa}), Problem~\ref{penalized} and Problem~\ref{equimmkp} have the same optimal solution.
\end{lemma}
\begin{proof}
Lemma~\ref{exact} can be obtained directly from Theorem~1 in \cite{exactpenalty}.
\end{proof}

Lemma~\ref{exact} illustrates that Problem~\ref{penalized} is equivalent to Problem~\ref{equimmkp} if the penalty parameter $\mu$ is sufficiently large. Thus, we can solve Problem~\ref{penalized} instead of Problem~\ref{equimmkp} by using CCCP.  However, it may not always be a feasible solution to Problem~\ref{equimmkp}. In order to obtain a global optima of Problem~\ref{equimmkp}, we obtain multiple local optimal solutions of Problem~\ref{penalized} via performing CCCP multiple times, each with a unique initial feasible point of Problem~\ref{penalized}, and then choose the one which achieves the minimum average value \cite{infeasible}.
\begin{figure}[t]
\begin{center}
 \subfigure[Cache size.]
  {\resizebox{8.3cm}{!}{\includegraphics{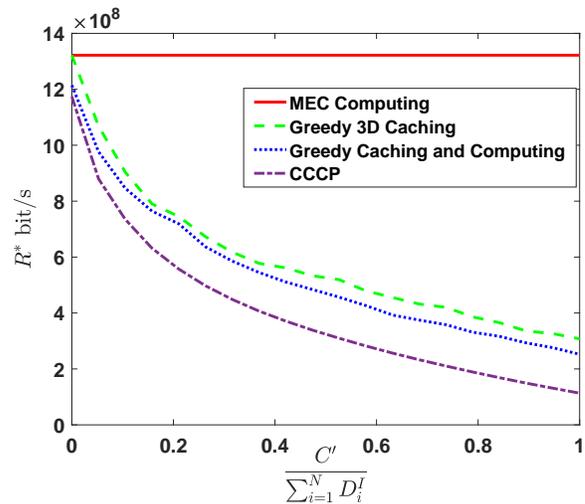}}}\quad\quad
 \subfigure[Computation frequency under CCCP.]
 {\resizebox{8.3cm}{!}{\includegraphics{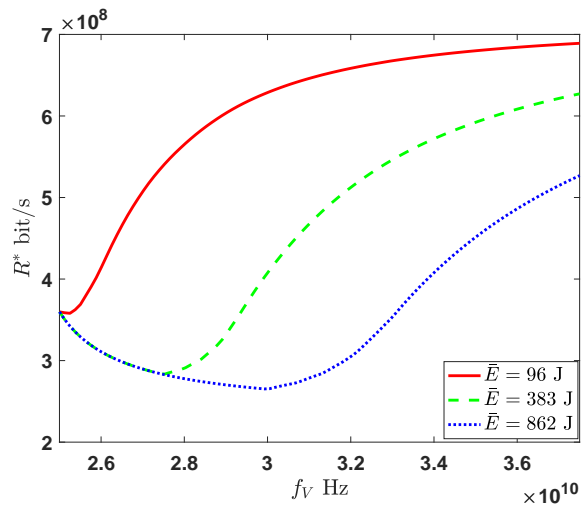}}}\quad\quad
\end{center}
   \caption{\small{Heterogeneous scenario analysis at $f_V= 50$ G Hz, $N=100, D_i^I \in [1,25]$ M bits, $\alpha =2$, $w=10$ cycle/bit, $P_i \propto \frac{1}{i^\gamma}$ with $\gamma= 0.8$, $C' = 0.3* \sum_{i=1}^ND_i^I$, $\bar{E} = 0.25*kwf_V^2\sum_{i=1}^NP_iD_i^I$, $\mu = 10^5$ unless otherwise stated}.}
\label{Effect22}
\end{figure}
\subsection{Numerical Results}

In this section, we illustrate the performance of CCCP via numerical results, as shown in Fig.~\ref{Effect22}. Specifically, CCCP is obtained via performing Algorithm~\ref{dca} with $\delta = 0.001$ 100 times, each starting with a random initial feasible point, and then selecting the local optima with the lowest average transmission rate value. We compare it with the following three baselines:
\begin{itemize}
\item MEC computing: requests for all FOVs are computed at the MEC server, i.e., $c_i^O=0,c_i^I=0,d_i=0$ for all $i\in \mathcal{N}$;
\item Greedy 3D caching: 3D FOVs are cached at the mobile VR device via greedy algorithm, i.e., sort $\mathcal{N}$ according to $\frac{P_iR_i^S}{D_i^O}$ in descending order, denote with $\lfloor j\rfloor$ the index $i\in \mathcal{N}$ with the $j$-th maximal value of $\frac{P_iR_i^S}{D_i^O}$, and $s_c$ the split index satisfying $\sum_{j=1}^{s_c-1} D^O_{\lfloor j \rfloor} \leq C$ and $\sum_{j=1}^{s_c} D^O_{\lfloor j \rfloor} > C$. Set $c_{\lfloor j\rfloor}^O=1,c_{\lfloor j\rfloor}^I=0,d_{\lfloor j\rfloor}=0$ for all $j\in \{1,\cdots,s_c\}$ and $c_{\lfloor j\rfloor}^O=0,c_{\lfloor j\rfloor}^I=0,d_{\lfloor j\rfloor}=0$, otherwise;
\item Greedy caching and computing: first, local computing with local 2D caching is determined via greedy algorithm, i.e., sort $\mathcal{N}$ according to $\frac{P_iR_i^S}{D_i^I+P_ikD_i^Iw_if_V^2}$ in descending order, denote with $\lfloor j\rfloor$ the index $i\in \mathcal{N}$ with the $j$-th maximal value of $\frac{P_iR_i^S}{D_i^I+P_ikD_i^Iw_if_V^2}$, and $s_c^1$ the split index satisfying $\sum_{j=1}^{s_c^1-1} D^I_{\lfloor j \rfloor} \leq C$ and $\sum_{j=1}^{s_c^1} D^I_{\lfloor j \rfloor} > C$ or $\sum_{j=1}^{s_c^1-1} P_{\lfloor j\rfloor}kD^I_{\lfloor j \rfloor}w_{\lfloor j\rfloor}f_V^2 \leq \bar{E}$ and $\sum_{j=1}^{s_c^1} P_{\lfloor j\rfloor}kD^I_{\lfloor j \rfloor}w_{\lfloor j\rfloor}f_V^2 > \bar{E}$. Set $c_{\lfloor j\rfloor}^O=0,c_{\lfloor j\rfloor}^I=1,d_{\lfloor j\rfloor}=1$ for all $j\in \{1,\cdots,s_c^1-1\}$ and $c_{\lfloor j\rfloor}^O=0,c_{\lfloor j\rfloor}^I=0,d_{\lfloor j\rfloor}=0$, otherwise; secondly, if there still exists underutilized cache size, i.e., $\sum_{j=1}^{s_c^1-1} D^I_{\lfloor j \rfloor} < C$, then 3D FOVs of the rest of viewpoints are cached at the mobile VR device via greedy algorithm. Otherwise, if $\sum_{j=1}^{s_c^1-1} P_{\lfloor j\rfloor}kD^I_{\lfloor j \rfloor}w_{\lfloor j\rfloor}f_V^2 < \bar{E}$, then local computing without caching is decided via greedy algorithm according to $\frac{P_i(R_i^S-R_i^V)}{P_ikD_i^Iw_if_V^2}$.
\end{itemize}


Fig.~\ref{Effect22}~(a) and Fig.~\ref{Effect22}~(b) illustrate the impacts of the local cache size, i.e., $C'$, and the local computation frequency, i.e., $f_V$, on the optimal average transmission rate in the heterogeneous scenario. We see that CCCP exhibits great promises in saving  communication bandwidth compared with the baselines. For example,  compared with greedy 3D caching and greedy caching and computing, CCCP brings larger transmission rate gain over MEC computing (e.g., 45\%,  48 \% vs. 63\% at $\frac{C'}{\sum_{i=1}^ND_i^I} = 20\%$).  

\section{Conclusion }

In this paper, we develop a novel MEC-based mobile VR delivery framework by jointly utilizing the caching and computing capacities of the mobile VR device. When FOVs are homogeneous, a closed-form expression for the optimal joint policy is derived, which reveals a fundamental tradeoff between the three primary resouces, i.e., communications, caching and computing. The tradeoff results show that:
\begin{itemize}
  \item When $f_V < F$, $R^*$ increases with $f_V$ if $\frac{N\bar{E}}{k f_V^2D^Iw} \leq C$ and stays unchanged with $f_V$, otherwise; $R^*$ decreases with $C$ at the rate of $\frac{R_S}{\alpha N}$ when $\frac{N\bar{E}}{k f_V^2D^Iw} \leq C$ and $\frac{R_S}{N}$, otherwise;
  \item When $F \leq f_V $, $R^*$ first decreases and then increases with $f_V$ if $\frac{N\bar{E}}{k f_V^2D^Iw} > C$ and increases with $f_V$, otherwise; $R^*$ decreases with $C$ at the rate of $\frac{R_S}{\alpha N}$ when $\frac{N\bar{E}}{k f_V^2D^Iw} \leq C$ and $\frac{R_V}{N}$, otherwise;
\end{itemize}

In the heterogeneous scenario, we transform the NP-hard problem into an equivalent IQP and solve it via CCCP, which obtains a local optima and is shown to achieve good performance in numerical results.

\section*{Appendix A: Proof of Lemma~\ref{ror1}}\label{proof of ror1}
When $f_V < F$, $R_S-R_V<0$ and the objective function of Problem~\ref{Prob2} increases with $d - \min\{c^I,d\}$. Thus, we can see that $d - \min\{c^I,d\} = 0$, i.e., $d\leq c^I$. In addition, based on Property~\ref{strucC}, by replacing $c^O$ with $ \frac{C - c^I}{\alpha}$, Problem~\ref{Prob2} can be rewritten as

\begin{align}
\textit{Problem 7}:
& \min_{c^I,d}  \ \ \ \ R_S\left(1-\frac{C}{\alpha N}\right) + \frac{R_S}{\alpha N}c^I-\frac{R_S}{N}d\nonumber\\
&\  s.t.\ \ \ \ \ c^I \in \left\{0,1,\cdots,C\right\},\label{c}\\
&\ \ \ \ d\in \left\{0,1,\cdots,\min\left\{c^I, \frac{N\bar{E}}{k f_V^2D^Iw}\right\}\right\}.\label{d}
\end{align}
In the following, we analyze the optimal solution to Problem~7 from the following two aspects.
\begin{itemize}
\item If $c^I \leq \frac{N\bar{E}}{k f_V^2D^Iw}$, (\ref{c}) and (\ref{d}) can be rewritten as
\begin{equation}\label{c1}
c^I \in \left\{0,1,\cdots,\min\left\{C,\frac{N\bar{E}}{k f_V^2D^Iw}\right\}\right\},
\end{equation}
\begin{equation}\label{d1}
d\in \left\{0,1,\cdots,c^I\right\}.
\end{equation}
Since the objective function of Problem~7 decreases with $d$, we have $d=c^I$ without loss of optimality. By replacing $d$ with $c^I$, and (\ref{c}) with (\ref{c1}), Problem~7 can be rewritten as
\begin{align}
\textit{Problem 8}:
& \min_{c^I}  \ \ \ \ R_S\left(1-\frac{C}{\alpha N}\right) - (\alpha-1)\frac{R_S}{\alpha N}c^I\nonumber\\
&\  s.t.\ \ c^I \in \left\{0,1,\cdots,\min\left\{C,\frac{N\bar{E}}{k f_V^2D^Iw}\right\}\right\}. \nonumber
\end{align}
Since $\alpha>1$, we can see that the objective function of Problem~8 decreases with $c^I$, and thus $c^{I^*} = \min\left\{C,\frac{N\bar{E}}{k f_V^2D^Iw}\right\}$. Accordingly, we have $d^*=\min\left\{C,\frac{N\bar{E}}{k f_V^2D^Iw}\right\}$ and  $c^{O^*} = \frac{C - c^{I^*}}{\alpha}$.
\item If $c^I \geq \frac{N\bar{E}}{k f_V^2D^Iw}$, (\ref{c}) and (\ref{d}) can be rewritten as
\begin{equation}\label{c11}
c^I \in \left\{\frac{N\bar{E}}{k f_V^2D^Iw},\cdots,C\right\},
\end{equation}
\begin{equation}\label{d11}
\ \ d\in \left\{0,1,\cdots,\frac{N\bar{E}}{k f_V^2D^Iw}\right\}.
\end{equation}
Since the objective function of Problem~7 decreases with $d$ and increases with $c^I$, we have $d^*=\frac{N\bar{E}}{k f_V^2D^Iw}$ and $c^{I^*} = \frac{N\bar{E}}{k f_V^2D^Iw}$. Accordingly, we have $c^{O^*}$ via $c^{O^*} = \frac{C - c^{I^*}}{\alpha}$. Since $c^I \geq \frac{N\bar{E}}{k f_V^2D^Iw}$ holds only when $C \geq \frac{N\bar{E}}{k f_V^2D^Iw}$, we have $c^{I^*} = \min\left\{C,\frac{N\bar{E}}{k f_V^2D^Iw}\right\}$, $d^*=\min\left\{C,\frac{N\bar{E}}{k f_V^2D^Iw}\right\}$ and  $c^{O^*} = \frac{C - c^{I^*}}{\alpha}$.
\end{itemize}
Thus, (\ref{3D1}), (\ref{2D1}) and (\ref{off1}) hold. By substituting (\ref{3D1}), (\ref{2D1}) and (\ref{off1}) into the objective function of Problem~\ref{Prob2}, (\ref{rate1}) holds. The proof ends.
\section*{Appendix B: Proof of Lemma~\ref{ror2}}\label{proof of ror2}
When $F \leq f_V$, $R_V\leq R_S<\alpha R_V$. We analyze the optimal solution to Problem~\ref{Prob2} from the following two aspects.
\begin{itemize}
\item If $c^I \leq d$, Problem~\ref{Prob2} can be rewritten as
\begin{align}
& \min_{c^I,d}  \ R_S\left(1-\frac{C}{\alpha N}\right) - \frac{\alpha R_V-R_S}{\alpha N}c^I-\frac{R_S-R_V}{N}d\nonumber\\
&\  s.t.\ \ \ \ \ \ \ \ \ \ \ \ \ \ \ \ \ c^I \in \left\{0,1,\cdots,\min\left\{d,C\right\}\right\},\label{c22}\\
&\ \ \ \ \ \ \  \ \ \ \ \ \ \ \ \ \ \ \ \ \ \  d\in \left\{0,1,\cdots,\frac{N\bar{E}}{k f_V^2D^Iw}\right\}.\label{d22}
\end{align}\label{Prob10}
Since $R_V\leq R_S<\alpha R_V$, we have the objective function decreases with $c^I$ and $d$. Thus, we have $c^{I^*} = \min\left\{d^*,C\right\}$, $d^*=\frac{N\bar{E}}{k f_V^2D^Iw}$ and $c^{O^*} = \frac{C - c^{I^*}}{\alpha}$.
\item If $c^I \geq d$, Problem~\ref{Prob2} can be rewritten as
\begin{align}
& \min_{c^I,d}  \ \ \ \ R_S\left(1-\frac{C}{\alpha N}\right) + \frac{R_S}{\alpha N}c^I-\frac{R_S}{N}d\nonumber\\
&\  s.t.\ \ \ \ \ \ \ c^I \in \left\{d,\cdots,C\right\},\label{cC}\\
&\ \ \ \ \ \ \  \ \ \ \ \   d\in \left\{0,1,\cdots,\frac{N\bar{E}}{k f_V^2D^Iw}\right\}.\label{dD}
\end{align}\label{Prob11}
Since  the objective function increases with $c^I$ and decreases with $d$, we have $c^{I^*} = d^*$, $d^*=\frac{N\bar{E}}{k f_V^2D^Iw}$ and $c^{O^*} = \frac{C - c^{I^*}}{\alpha}$. In addition, since $c^{I^*} = d^*$ holds only when $C\geq d^*$, $c^{I^*}$ can also be rewritten as $c^{I^*} = \min\left\{d^*, C\right\}$.
\end{itemize}
Thus, (\ref{3D2}), (\ref{2D2}) and (\ref{off2}) hold. By substituting (\ref{3D2}), (\ref{2D2}) and (\ref{off2}) into the objective function of Problem~\ref{Prob2}, (\ref{rate2}) holds. The proof ends.
\bibliographystyle{IEEEtran}

\begin{thebibliography}{}
\providecommand{\url}[1]{#1}
\csname url@samestyle\endcsname
\providecommand{\newblock}{\relax}
\providecommand{\bibinfo}[2]{#2}
\providecommand{\BIBentrySTDinterwordspacing}{\spaceskip=0pt\relax}
\providecommand{\BIBentryALTinterwordstretchfactor}{4}
\providecommand{\BIBentryALTinterwordspacing}{\spaceskip=\fontdimen2\font plus
\BIBentryALTinterwordstretchfactor\fontdimen3\font minus
  \fontdimen4\font\relax}
\providecommand{\BIBforeignlanguage}[2]{{%
\expandafter\ifx\csname l@#1\endcsname\relax
\typeout{** WARNING: IEEEtran.bst: No hyphenation pattern has been}%
\typeout{** loaded for the language `#1'. Using the pattern for}%
\typeout{** the default language instead.}%
\else
\language=\csname l@#1\endcsname
\fi
#2}}
\providecommand{\BIBdecl}{\relax}
\BIBdecl

\end{thebibliography}


\begin{thebibliography}{1}
\providecommand{\url}[1]{#1}
\csname url@samestyle\endcsname
\providecommand{\newblock}{\relax}
\providecommand{\bibinfo}[2]{#2}
\providecommand{\BIBentrySTDinterwordspacing}{\spaceskip=0pt\relax}
\providecommand{\BIBentryALTinterwordstretchfactor}{4}
\providecommand{\BIBentryALTinterwordspacing}{\spaceskip=\fontdiMEN2\font plus
\BIBentryALTinterwordstretchfactor\fontdiMEN3\font minus
  \fontdiMEN4\font\relax}
\providecommand{\BIBforeignlanguage}[2]{{%
\expandafter\ifx\csname l@#1\endcsname\relax
\typeout{** WARNING: Ieeetran.bst: No hyphenation pattern has been}%
\typeout{** loaded for the language `#1'. Using the pattern for}%
\typeout{** the default language instead.}%
\else
\language=\csname l@#1\endcsname
\fi
#2}}
\providecommand{\BIBdecl}{\relax}
\BIBdecl
\bibitem{Sunvr}
Y.~Sun, Z.~Chen, M.~Tao, and H.~Liu, ``Communication, Computing and Caching for Mobile VR Delivery: Modeling and Trade-off,'' \emph{IEEE ICC}, to appear, May 2018.
\bibitem{burst}
R.~Begole, [Online] available: https://www.forbes.com/sites/valleyvoices/\\2016/02/09/why-the-internet-pipes-will-burst-if-virtual-reality-takes-off/.
\bibitem{whitepaper}
ABI Research, Qualcomm, ``Augmented and Virtual Reality: the First Wave of 5G Killer Apps,'' [Online] available: https://www.qualcomm.com/documents/augmented-and-virtual-reality-first-wave-5g-killer-apps, Feb. 2017.
\bibitem{VRmarket}
Juniper, \!``Virtual reality markets: Hardware, content $\&$ accessories 2017--2022,'' [Online] available: https://www.juniperresearch.com/researchstore/innovation-disruption/virtual-reality/hardware-content-accessories, 2017.
\bibitem{e}
E.~Bastug, M.~Bennis, M.~Medard and M.~Debbah, ``Toward Interconnected Virtual Reality: Opportunities, Challenges, and Enablers,'' \emph{IEEE Commun. Mag.}, vol.~55, no.~6, pp.~110--117, June. 2017.
\bibitem{Simone}
S. Mangiante, G.~Klas, A.~Navon, G.~Zhuang, J.~Ran and M.~D.~Silva, ``VR is on the Edge: How to Deliver $360^{\circ}$ Videos in Mobile Networks,'' \emph{ACM SIGCOMM. Workshop on VR/AR Network}, Aug. 2017.
\bibitem{3D}
S.~Reichelt, R.~Hausselr, G.~Futterer, and N.~Leister, ``Depth cues in human visual perception and their realization in 3D displays,'' \emph{in SPIE Three-Dimensional Imaging, Visualization, and Display 2010}, Orlando, FL, April 2010.
\bibitem{localoptima}
Le~Thi Hoai An and Pham Dinh Tao, ``Solving a class of linearly constrained indefinite quadratic problems by D.C. algorithms,'' \emph{Journal of Global Optimization}, vol. 11, no. 3, pp. 253--285, 1997.
\bibitem{tile1}
V.~R.~Gaddam, M.~Riegler, R.~Eg, C.~Griwodz, and P.~Halvorsen, ``Tiling in interactive panoramic video: Approaches and Evaluation,'' \emph{IEEE Trans. on Multimedia}, vol.~18, no.~9, pp.~1819--1831, 2016.
\bibitem{tile2}
T.~El-Ganainy and M.~Hefeeda. 2016. Streaming Virtual Reality Content. arXiv preprint arXiv:1612.08350 (2016).
\bibitem{qian}
F.~Qian, L.~Ji, R.~Han, and V.~Gopalakrishnan, ``Optimizing 360 video delivery over cellular networks,'' \emph{in Proceedings of the 5th Workshop on All Things Cellular: Operations, Applications and Challenges. ACM}, pp.~1--6, 2016.
\bibitem{fixation}
C.~Fan, J.~Lee, W.~Lo, C.~Huang, K.~Chen, and C.~Hsu, ``Fixation Prediction for 360∼ Video Streaming to Head-Mounted Displays,'' \emph{in Proceedings of ACM NOSSDAV}, June 2017.
\bibitem{chen1}
M. Chen, W. Saad and C. Yin, ``Resource Management for Wireless Virtual Reality: Machine Learning Meets Multi-Attribute Utility,''  \emph{IEEE GLOBECOM}, Singapore, Dec. 2017, pp. 1-7.
\bibitem{chen2}
M. Chen, W. Saad and C. Yin, ``Virtual Reality over Wireless Networks: Quality-of-Service Model and Learning-Based Resource Management,'' \emph{IEEE Trans. on Commun.,} to appear, 2018.
\bibitem{3C}
H.~Liu, Z.~Chen, and L.~Qian, ``The three primary colors of mobile systems,'' \emph{IEEE Commun. Mag.}, vol.~54, no.~9, pp.~15--21, Sep.~2016.
\bibitem{3C1}
H.~Liu, Z.~Chen, X.~Tian, X.~Wang, and M.~Tao, ``On content-centric wireless delivery networks,'' \emph{IEEE Wireless Commun.}, vol.~21, no.~6, pp.~118--125, Jan.~2014.
\bibitem{erol}
Erol-Kantarci M., Sukhmani~S., ``Caching and Computing at the Edge for Mobile Augmented Reality and Virtual Reality (AR/VR) in 5G,''  \emph{Ad Hoc Networks.}, vol.~223, pp.~169--177, 2018.
\bibitem{mohammed1}
M.~S.~Elbamby, C.~Perfecto, M.~Bennis and K.~Doppler, ``Towards Low-Latency and Ultra-Reliable Virtual Reality,'' \emph{in IEEE Network}, vol. 32, no. 2, pp. 78-84, 2018. 
\bibitem{colla}
A.~Ndikumana, S.~Ullah, T.~LeAnh, N.~H.~Tran, and C.~S.~Hong, ``Collaborative Cache Allocation and Computation Offloading in Mobile Edge Computing,'' \emph{IEEE Asia-Pacific Netw. Operation Manage. Symp.}, pp. 366--369, Sep.~2017.
\bibitem{bigdata}
A.~Ndikumana, N.~H.~Tran, T.~M.~Ho, Z.~Han, W.~Saad, D.~Niyato and C.~S.~Hong, ``Joint Communication, Computation, Caching, and Control in Big Data Multi-access Edge Computing,'' arXiv preprint arXiv:1803.11512, 2018.
\bibitem{game}
S.~Kim, ``5G Network Communication, Caching, and Computing Algorithms Based on the Two--Tier Game Model,'' \emph{ETRI Journal}, vol.~40, no.~1, pp.~61--71, Feb.~2018.
\bibitem{J}
J.~Chakareski, ``VR/AR Immersive Communication: Caching, Edge Computing, and Transmission Trade-Offs,'' \emph{ACM SIGCOMM. Workshop on VR/AR Network}, Aug.~2017.
\bibitem{mohammed2}
M.~S.~Elbamby, C.~Perfecto, M.~Bennis and K.~Doppler, ``Edge Computing Meets Millimeter-wave Enabled VR: Paving the Way to Cutting the Cord,'' \emph{IEEE WCNC}, 2018.
\bibitem{xiaoyang}
X.~Yang, Z.~Chen, K.~Li, Y.~Sun, N.~Liu, W.~Xie and Y.~Zhao, ``Communication--Constrained Mobile Edge Computing Systems for Wireless Virtual Reality: Scheduling and Tradeoff,'' \emph{IEEE Access}, vol.~6, pp.~16665--16677, Mar.~2018.


\bibitem{cycle}
Y.~Mao, C.~You, J.~Zhang, K.~Huang, and K.~R.~Letaief, ``A Survey on Mobile Edge Computing: The Communication Perspective,'' \emph{IEEE Commun. Surveys $\&$ Tutorials}, vol.~19, no.~4, pp.~2322--2358, 2017.
\bibitem{mmkp}
M.~Hifi, M.~Michrafy and A.~Sbihi, ``Heuristic algorithms for the multiple-choice multidimensional knapsack problem,'' \emph{Journal of the Operational Research Society}, vol.~55, pp. 1323--1332, 2004.
\bibitem{exactpenalty}
H.~A. Le~Thi, T. P. Dinh, and H. Van Ngai, ``Exact penalty and error bounds in dc programming,'' \emph{Journal of Global Optimization}, vol. 52, no. 3, pp. 509--535, 2012.
\bibitem{infeasible}
H. A. Le Thi, T. P. Dinh, H. M. Le, and X. T. Vo, ``DC approximation approaches for sparse optimization,'' \emph{european Journal of Operational Research}, vol.~244, no.~1, pp.~26--46,~2015.
\bibitem{dcacccp}
 B.~Sriperumbudur and G.~Lanckriet, ``On the convergence of concave-convex procedure,'' \emph{NIPS Workshop on Optimization for Machine Learning}, 2009.
\end{thebibliography}

\end{document}